\documentclass{article}
\usepackage{amssymb}
\usepackage{amsfonts}
\usepackage{amsmath}

\newtheorem{theorem}{Theorem}
\newtheorem{acknowledgement}[theorem]{Acknowledgement}

\newtheorem{definition}[theorem]{Definition}

\newtheorem{remark}[theorem]{Remark}

\newenvironment{proof}[1][Proof]{\noindent\textbf{#1.} }{\ \rule{0.5em}{0.5em}}
\input{tcilatex}

\begin{document}

\title{A General Variational Principle of Classical Field and Its
Application to General relativity I}
\author{Zhaoyan Wu \\
Center for Theoretical Physics, Jilin University}
\maketitle

\begin{abstract}
A general variational principle of classical fields with a Lagrangian
containing field quantity and its derivatives of up to the N-order is
presented. Noether's theorem is derived. The generalized Hamilton-Jacobi's
equation for the Hamilton's principal functional is obtained. These results
are surprisingly in great harmony with each other. They will be applied to
general relativity in the subsequent articles, especially the generalized
Noether's theorem will be applied to the problem of conservation and
non-conservation in curved spacetime.
\end{abstract}

\section{Introduction}

The aim of this series of articles is to explore the conservation and
non-conservation in curved spacetime, especially to explore the difficulty
of energy-momentum conservation in general relativity and the gravitational
energy-momentum. We start with presenting a general variational principle
for classical fields with a Lagrangian containing the field quantity and its
derivatives of up to the $N$-th order (part I), Then the general results
from part I are applied to general relativity, especially the generalized
Noether's theorem is applied to the problem of conservation and
non-conservation in general relativity (part II). The last part (part III)
is devoted to the difficulty of conservation of energy-momentum in curved
spacetime and to the problem whether the metric field carries
energy-momentum or not.

The\ developments of modern physics, such as the founding of statistical
mechanics and quantum mechanics, have proved that the variational principle
approach to dynamics is not only an alternative and equivalent version to
the naive, intuitive approach, but also yields deeper insights into the
underlying physics. For instance, it is hard to imagine that the statistical
mechanics could have been established\ without using the concepts of phase
space, and the quantum mechanics\ could have been established without using
the concept of Hamiltonian. Therefore, we will found our argument on a
general variational principle for classical field. It might be for the same
reason, soon after Einstein proposed his general theory of relativity,
Hilbert made the first attempt to get Einstein's equation by using the least
action principle. The Lagrangian being used for vacuum Einstein's equation, $%
(16\pi G)^{-1}R$, is the only independent scalar constructed in terms of the
metric field and its derivatives of no higher than the second order.
However, because the Ricci scalar curvature $R$ contains the second order
derivatives of the metric field $g_{\mu \nu }(x)$, which is now the dynamic
variable, the least action priciple for Lagrangians containing only the
field quantity and its first order derivatives\ does not lead to Einstein's
field equation. The generally accepted solution to this difficulty is adding
the Gibbons-Hawking boundary term to the Hilbert action and keeping the
least action principle unchanged[1]. But there is another solution to this
difficulty, which is adopted in the present paper. The least action
principle will be restated and the Hilbert action will still be used for the
vacuum Einstein's equation. In order to show this is proper and natural, we
will consider classical fields with a Lagrangian containing the field
quantity and its derivatives of up to the $N$-th order. In our opinion,
acting at a distance is not acceptable, so we assume that the Lagrangian
does not contain the integral of the field quantity. In section 2, a general
Lagrangian formalism for classical fields is presented. In section 3, the
Hamiltonian formalism is discussed. In section 4, the Noether's theorem is
derived and\ the conservation law due to the "coordinate shift" invariance
is established. The generalized Hamilton-Jacobi's equation is obtained in
section 5. All the results obtained above are in great harmony with each
other and apply to various classical fields, say, those with Galilean
covariance, with Lorentzian covariance, with general covariance, or without
such covariance. Part I finishes with a remark. In part II, this general
variational\ principle of classical fields developed in part I is applied to
general relativity and\ quite a few conserved quantities corresponding to
the coordinate "shift" invariance, coordinate "rotation" invariance etc. are
found. And the properties of these conservation laws are discussed. In part
III, after some general consideration, the introducing of gravitational
energy-momentum is reviewed. The difficulties of conservation of
energy-momentum in general relativity are explored by using Noether's
theorem and observations from geometry. It is pointed out that the metric
field does not carry energy-momentum, and the law of conservation of
energy-momentum no longer holds in curved spacetime.

\section{Lagrangian formulation of classical fields}

Suppose that our spacetime $M$ is a smooth manifold which is differentially
homeomorphic to $\mathbb{R}^{4}$. Choose a chart $(M,\varphi )$, and denote
by $(x^{0},x^{1},x^{2},x^{3})$ the corresponding coordinates. Suppose the
action over any spacetime region $\Omega \subset M$ of the classical field $%
\{\Phi _{a}(x)\}$ is 
\begin{equation}
A=\int_{\Omega }d^{4}xL(x,\Phi (x),\partial \Phi (x),\partial ^{2}\Phi
(x),\ldots ,\partial ^{N}\Phi (x))=A[\Phi ]  \tag{1}
\end{equation}%
where the Lagrangian $L$ is a function of the spacetime coordinate, the
field $\{\Phi _{a}(x)\}$ and its derivatives of no higher than the $N$-th
order. $L$ does not contain the integral of $\{\Phi _{a}(x)\}$, since acting
at a distance\ is not acceptable. In order to develop a general variational
principle for all locally interacting classical fields, here we suppose that 
$L$ can manifestly contain the spacetime coordinates, and the Galilean
invariance,\ Lorentzian invariance or the general invariance are not assumed
for the time being. For the sake of simplicity, we have assumed that our
spacetime\ manifold is $(3+1)$-dimensional. However, our presentation has
nothing to do with the spacetime dimensionality. It still holds for an $%
(n+1) $-dimensional spacetime.

Consider the difference between actions over $\Omega$ of two possible
movements close to each other. Using integration by parts and Stokes
theorem, one gets

\begin{align}
\delta A[\Phi] & =\int_{\Omega}d^{4}x\delta\Phi_{a}(x)[\frac{\partial L}{%
\partial\Phi_{a}(x)}-\partial_{\lambda_{1}}\frac{\partial L}{\partial
\partial_{\lambda_{1}}\Phi_{a}(x)}+-\cdots  \notag \\
& +(-1)^{N}\partial_{\lambda_{1}}\cdots\partial_{\lambda_{N}}\frac{\partial L%
}{\partial\partial_{\lambda_{1}}\cdots\partial_{\lambda_{N}}\Phi_{a}(x)}%
]+\int_{\partial\Omega}ds_{\lambda}[B^{a\lambda}\delta\Phi_{a}(x)+B^{a%
\lambda\nu_{1}}\delta\partial_{\nu_{1}}\Phi_{a}(x)  \notag \\
&
+B^{a\lambda\nu_{1}\nu_{2}}\delta\partial_{\nu_{1}}\partial_{\nu_{2}}%
\Phi_{a}(x)+\cdots+B^{a\lambda\nu_{1}\cdots\nu_{N-1}}\delta\partial_{%
\nu_{1}}\cdots\partial_{\nu_{N-1}}\Phi_{a}(x)],  \tag{2}
\end{align}
where the Greek indices go through $0,1,2,3,$ and

\begin{align}
B^{a\lambda} & =\frac{\partial L}{\partial\partial_{\lambda}\Phi_{a}(x)}%
-\partial_{\mu_{1}}\frac{\partial L}{\partial\partial_{\lambda}\partial_{%
\mu_{1}}\Phi_{a}(x)}+\partial_{\mu_{1}}\partial_{\mu_{2}}\frac{\partial L}{%
\partial\partial_{\lambda}\partial_{\mu_{1}}\partial _{\mu_{2}}\Phi_{a}(x)}%
-+\cdots  \notag \\
& +(-1)^{N-1}\partial_{\mu_{1}}\cdots\partial_{\mu_{N-1}}\frac{\partial L}{%
\partial\partial_{\lambda}\partial_{\mu_{1}}\cdots\partial_{\mu_{N-1}}%
\Phi_{a}(x)},  \notag \\
B^{a\lambda\nu_{1}} & =\frac{\partial L}{\partial\partial_{\lambda}%
\partial_{\nu_{1}}\Phi_{a}(x)}-\partial_{\mu_{1}}\frac{\partial L}{%
\partial\partial_{\lambda}\partial_{\mu_{1}}\partial_{\nu_{1}}\Phi_{a}(x)}%
+\partial_{\mu_{1}}\partial_{\mu_{2}}\frac{\partial L}{\partial
\partial_{\lambda}\partial_{\mu_{1}}\partial_{\mu_{2}}\partial_{\nu_{1}}%
\Phi_{a}(x)}-+\cdots  \notag \\
& +(-1)^{N-2}\partial_{\mu_{1}}\cdots\partial_{\mu_{N-2}}\frac{\partial L}{%
\partial\partial_{\lambda}\partial_{\mu_{1}}\cdots\partial_{\mu_{N-2}}%
\partial_{\nu_{1}}\Phi_{a}(x)},\ldots,  \notag \\
B^{a\lambda\nu_{1}\cdots\nu_{N-2}} & =\frac{\partial L}{\partial
\partial_{\lambda}\partial_{\nu_{1}}\cdots\partial_{\nu_{N-2}}\Phi_{a}(x)}%
-\partial_{\mu_{1}}\frac{\partial L}{\partial\partial_{\lambda}\partial_{%
\mu_{1}}\partial_{\nu_{1}}\cdots\partial_{\nu_{N-2}}\Phi_{a}(x)},  \notag \\
B^{a\lambda\nu_{1}\cdots\nu_{N-1}} & =\frac{\partial L}{\partial
\partial_{\lambda}\partial_{\nu_{1}}\cdots\partial_{\nu_{N-1}}\Phi_{a}(x)}. 
\tag{3}
\end{align}

Equation\ (2) suggests \textbf{the least action principle} reads as follows.

For \textit{any spacetime region }$\Omega$, a\textit{mong all possible
movements in }$\Omega$ \textit{with the same boundary condition}

\begin{equation}
\delta\Phi|_{\partial\Omega}=0,\delta\partial\Phi|_{\partial\Omega}=0,%
\ldots,\delta\partial^{N-1}\Phi|_{\partial\Omega}=0,  \tag{4}
\end{equation}

\textit{the real movement corresponds to the stationary value of the action
over }$\Omega$\textit{.}

Combining eqns.(2), (4), one obtains\textbf{\ the field equation (
Euler-Lagrange equation ) }satisfied\textbf{\ }by the real movement

\begin{equation}
\frac{\delta A}{\delta\Phi_{a}(x)}=\frac{\partial L}{\partial\Phi_{a}(x)}%
-\partial_{\lambda_{1}}\frac{\partial L}{\partial\partial_{\lambda_{1}}%
\Phi_{a}(x)}+-\cdots+(-1)^{N}\partial_{\lambda_{1}}\cdots\partial
_{\lambda_{N}}\frac{\partial L}{\partial\partial_{\lambda_{1}}\cdots
\partial_{\lambda_{N}}\Phi_{a}(x)}=0.  \tag{5}
\end{equation}

\section{Hamiltonian formulation of classical fields}

To formulate the Hamiltonian formalism of classic fields, one needs to
specify a reference coordinate system $(\xi^{0},\xi^{1},\xi^{2},\xi^{3})$,
such that the hyper-surfaces, $\Sigma_{\xi^{0}},$of constant $\xi^{0}$ are
spacelike Cauchy hyper-surfaces and the curves, $t_{\overrightarrow{\xi}}$
,of constant $\overrightarrow{\xi}$ are timelike world lines of the observer
at $\overrightarrow{\xi}$. We will consider the state of the field on
hyper-surfaces, $\Sigma_{\xi^{0}}$, and investigate the change of state
(evolution) with $\xi^{0}$. We will observe the state on $\Sigma_{\xi^{0}}$
and the evolution with $\xi^{0}$ from any reference coordinate system in the
same way. As has been pointed out[2], the properly formulated Hamiltonian
formalism is compatible with all dynamic systems, Galilean invariant,
Lorentzian invariant, general invariant and so on. The invariance is the
heritage from the Lagrangian being used. The $3+1$ decomposition of
spacetime proposed above is more general than the one generally accepted in
General relativity. The latter relies on the unknown dynamical variable, the
metric field. Suppose the Lagrangian functional is

\begin{equation}
\Lambda=\int d^{3}\xi L(\xi,\Phi(\xi),\partial\Phi(\xi),\ldots,\partial
^{N}\Phi(\xi))=\Lambda\lbrack\xi^{0},\Phi|_{\xi^{0}},\partial_{0}\Phi
|_{\xi^{0}},\ldots,\partial_{0}^{N}\Phi|_{\xi^{0}}],  \tag{6}
\end{equation}
Consider the difference between Lagrangian functionals of two states close
to each other. Using the general formula (2), which is independent of the
dimensionality, one gets

\begin{align*}
\delta \Lambda & =\int d^{3}\xi \{[\frac{\partial L}{\partial \Phi _{a}(\xi )%
}-\partial _{i_{1}}\frac{\partial L}{\partial \partial _{i_{1}}\Phi _{a}(\xi
)}+-\cdots  \\
& +(-1)^{N}\partial _{i_{1}}\cdots \partial _{i_{N}}\frac{\partial L}{%
\partial \partial _{i_{1}}\cdots \partial _{i_{N}}\Phi _{a}(\xi )}]\delta
\Phi _{a}(\xi ) \\
& +[\frac{\partial L}{\partial \partial _{0}\Phi _{a}(\xi )}%
-C_{2}^{1}\partial _{i_{1}}\frac{\partial L}{\partial \partial
_{i_{1}}\partial _{0}\Phi _{a}(\xi )}+-\cdots  \\
& +(-1)^{N-1}C_{N}^{1}\partial _{i_{1}}\cdots \partial _{i_{N-1}}\frac{%
\partial L}{\partial \partial _{i_{1}}\cdots \partial _{i_{N-1}}\partial
_{0}\Phi _{a}(\xi )}]\delta \partial _{0}\Phi _{a}(\xi ) \\
& +[\frac{\partial L}{\partial \partial _{0}^{2}\Phi _{a}(\xi )}%
-C_{3}^{2}\partial _{i_{1}}\frac{\partial L}{\partial \partial
_{i_{1}}\partial _{0}^{2}\Phi _{a}(\xi )}+-\cdots  \\
& +(-1)^{N-2}C_{N}^{2}\partial _{i_{1}}\cdots \partial _{i_{N-2}}\frac{%
\partial L}{\partial \partial _{i_{1}}\cdots \partial _{i_{N-2}}\partial
_{0}^{2}\Phi _{a}(\xi )}]\delta \partial _{0}^{2}\Phi _{a}(\xi )
\end{align*}%
\begin{eqnarray*}
&&+\cdots +[\frac{\partial L}{\partial \partial _{0}^{N-1}\Phi _{a}(\xi )}%
-C_{N}^{N-1}\partial _{i_{1}}\frac{\partial L}{\partial \partial
_{i_{1}}\partial _{0}^{N-1}\Phi _{a}(\xi )}]\delta \partial _{0}^{N-1}\Phi
_{a}(\xi ) \\
&&+\frac{\partial L}{\partial \partial _{0}^{N}\Phi _{a}(\xi )}\delta
\partial _{0}^{N}\Phi _{a}(\xi )\}+
\end{eqnarray*}%
\begin{align}
& +\int d^{3}\xi \partial _{i}\{[K^{ai}\delta \Phi _{a}(\xi
)+K^{aik_{1}}\delta \partial _{k_{1}}\Phi _{a}(\xi )+K^{aik_{1}k_{2}}\delta
\partial _{k_{1}}\partial _{k_{2}}\Phi _{a}(\xi )  \notag \\
& +\cdots +K^{aik_{1}\cdots k_{N-1}}\delta \partial _{k_{1}}\cdots \partial
_{k_{N-1}}\Phi _{a}(\xi )]  \notag \\
& +[K_{1}^{ai}\delta \partial _{0}\Phi _{a}(\xi )+K_{1}^{aik_{1}}\delta
\partial _{k_{1}}\partial _{0}\Phi _{a}(\xi )+K_{1}^{aik_{1}k_{2}}\delta
\partial _{k_{1}}\partial _{k_{2}}\partial _{0}\Phi _{a}(\xi )  \notag \\
& +\cdots +K_{1}^{aik_{1}\cdots k_{N-2}}\delta \partial _{k_{1}}\cdots
\partial _{k_{N-2}}\partial _{0}\Phi _{a}(\xi )]  \notag \\
& +[K_{2}^{ai}\delta \partial _{0}^{2}\Phi _{a}(\xi )+K_{2}^{aik_{1}}\delta
\partial _{k_{1}}\partial _{0}^{2}\Phi _{a}(\xi )+K_{2}^{aik_{1}k_{2}}\delta
\partial _{k_{1}}\partial _{k_{2}}\partial _{0}^{2}\Phi _{a}(\xi )  \notag \\
& +\cdots +K_{2}^{aik_{1}\cdots k_{N-3}}\delta \partial _{k_{1}}\cdots
\partial _{k_{N-3}}\partial _{0}^{2}\Phi _{a}(\xi )]  \notag \\
& +\cdots +[K_{N-2}^{ai}\delta \partial _{0}^{N-2}\Phi _{a}(\xi
)+K_{N-2}^{aik_{1}}\delta \partial _{k_{1}}\partial _{0}^{N-2}\Phi _{a}(\xi
)]+K_{N-1}^{ai}\delta \partial _{0}^{N-1}\Phi _{a}(\xi )\},  \tag{7}
\end{align}%
where the domain of integration is $\mathbb{R}^{3}$, the latin indices go
through $1,2,3,$and

\begin{align*}
K^{ai} & =\frac{\partial L}{\partial\partial_{i}\Phi_{a}(\xi)}%
-\partial_{j_{1}}\frac{\partial L}{\partial\partial_{i}\partial_{j_{1}}%
\Phi_{a}(\xi)}+\partial_{j_{1}}\partial_{j_{2}}\frac{\partial L}{%
\partial\partial_{i}\partial_{j_{1}}\partial_{j_{2}}\Phi_{a}(\xi)}-+\cdots \\
& +(-1)^{N-1}\partial_{j_{1}}\cdots\partial_{j_{N-1}}\frac{\partial L}{%
\partial\partial_{i}\partial_{j_{1}}\cdots\partial_{j_{N-1}}\Phi_{a}(\xi)},
\end{align*}%
\begin{align*}
K^{aik_{1}} & =\frac{\partial L}{\partial\partial_{i}\partial_{k_{1}}%
\Phi_{a}(\xi)}-\partial_{j_{1}}\frac{\partial L}{\partial\partial_{i}%
\partial_{j_{1}}\partial_{k_{1}}\Phi_{a}(\xi)}+\partial_{j_{1}}\partial
_{j_{2}}\frac{\partial L}{\partial\partial_{i}\partial_{j_{1}}%
\partial_{j_{2}}\partial_{k_{1}}\Phi_{a}(\xi)}-+\cdots \\
& +(-1)^{N-2}\partial_{j_{1}}\cdots\partial_{j_{N-2}}\frac{\partial L}{%
\partial\partial_{i}\partial_{j_{1}}\cdots\partial_{j_{N-2}}\partial
_{k_{1}}\Phi_{a}(\xi)},\ldots,
\end{align*}

\begin{align*}
K^{aik_{1}\cdots k_{N-2}} & =\frac{\partial L}{\partial\partial_{i}%
\partial_{k_{1}}\cdots\partial_{k_{N-2}}\Phi_{a}(\xi)}-\partial_{j_{1}}\frac{%
\partial L}{\partial\partial_{i}\partial_{j_{1}}\partial_{k_{1}}\cdots%
\partial_{k_{N-2}}\Phi_{a}(\xi)}, \\
K^{aik_{1}\cdots k_{N-1}} & =\frac{\partial L}{\partial\partial_{i}%
\partial_{k_{1}}\cdots\partial_{k_{N-1}}\Phi_{a}(\xi)},
\end{align*}%
\begin{align*}
K_{1}^{ai} & =c_{2}^{1}\frac{\partial L}{\partial\partial_{i}\partial
_{0}\Phi_{a}(\xi)}-c_{3}^{1}\partial_{j_{1}}\frac{\partial L}{\partial
\partial_{i}\partial_{j_{1}}\partial_{0}\Phi_{a}(\xi)}+c_{4}^{1}%
\partial_{j_{1}}\partial_{j_{2}}\frac{\partial L}{\partial\partial_{i}%
\partial_{j_{1}}\partial_{j_{2}}\partial_{0}\Phi_{a}(\xi)} \\
& -+\cdots+(-1)^{N-2}c_{N}^{1}\partial_{j_{1}}\cdots\partial_{j_{N-2}}\frac{%
\partial L}{\partial\partial_{i}\partial_{j_{1}}\cdots\partial_{j_{N-2}}%
\partial_{0}\Phi_{a}(\xi)},
\end{align*}%
\begin{align*}
K_{1}^{aik_{1}} & =c_{3}^{1}\frac{\partial L}{\partial\partial_{i}%
\partial_{k_{1}}\partial_{0}\Phi_{a}(\xi)}-c_{4}^{1}\partial_{j_{1}}\frac{%
\partial L}{\partial\partial_{i}\partial_{j_{1}}\partial_{k_{1}}\partial_{0}%
\Phi_{a}(\xi)}+c_{5}^{1}\partial_{j_{1}}\partial_{j_{2}}\frac{\partial L}{%
\partial\partial_{i}\partial_{j_{1}}\partial_{j_{2}}\partial_{k_{1}}%
\partial_{0}\Phi_{a}(\xi)}-+\cdots \\
& +(-1)^{N-3}c_{N}^{1}\partial_{j_{1}}\cdots\partial_{j_{N-3}}\frac{\partial
L}{\partial\partial_{i}\partial_{j_{1}}\cdots\partial_{j_{N-3}}\partial
_{k_{1}}\partial_{0}\Phi_{a}(\xi)},\ldots,
\end{align*}%
\begin{align*}
K_{1}^{aik_{1}\cdots k_{N-3}} & =c_{N-1}^{1}\frac{\partial L}{\partial
\partial_{i}\partial_{k_{1}}\cdots\partial_{k_{N-3}}\partial_{0}\Phi_{a}(\xi
)}-c_{N}^{1}\partial_{j_{1}}\frac{\partial L}{\partial\partial_{i}%
\partial_{j_{1}}\partial_{k_{1}}\cdots\partial_{k_{N-3}}\partial_{0}\Phi
_{a}(\xi)}, \\
K_{1}^{aik_{1}\cdots k_{N-2}} & =c_{N}^{1}\frac{\partial L}{\partial
\partial_{i}\partial_{k_{1}}\cdots\partial_{k_{N-2}}\partial_{0}\Phi_{a}(\xi
)},
\end{align*}%
\begin{align*}
K_{2}^{ai} & =c_{3}^{2}\frac{\partial L}{\partial\partial_{i}\partial
_{0}^{2}\Phi_{a}(\xi)}-c_{4}^{2}\partial_{j_{1}}\frac{\partial L}{%
\partial\partial_{i}\partial_{j_{1}}\partial_{0}^{2}\Phi_{a}(\xi)}%
+c_{5}^{2}\partial_{j_{1}}\partial_{j_{2}}\frac{\partial L}{\partial\partial
_{i}\partial_{j_{1}}\partial_{j_{2}}\partial_{0}^{2}\Phi_{a}(\xi)}-+\cdots \\
& +(-1)^{N-3}c_{N}^{2}\partial_{j_{1}}\cdots\partial_{j_{N-3}}\frac{\partial
L}{\partial\partial_{i}\partial_{j_{1}}\cdots\partial_{j_{N-3}}\partial
_{0}^{2}\Phi_{a}(\xi)},
\end{align*}%
\begin{align*}
K_{2}^{aik_{1}} & =c_{4}^{2}\frac{\partial L}{\partial\partial_{i}%
\partial_{k_{1}}\partial_{0}^{2}\Phi_{a}(\xi)}-c_{5}^{2}\partial_{j_{1}}%
\frac{\partial L}{\partial\partial_{i}\partial_{j_{1}}\partial_{k_{1}}%
\partial_{0}^{2}\Phi_{a}(\xi)}+c_{6}^{2}\partial_{j_{1}}\partial_{j_{2}}%
\frac{\partial L}{\partial\partial_{i}\partial_{j_{1}}\partial_{j_{2}}%
\partial_{k_{1}}\partial_{0}^{2}\Phi_{a}(\xi)}-+\cdots \\
& +(-1)^{N-4}c_{N}^{2}\partial_{j_{1}}\cdots\partial_{j_{N-4}}\frac{\partial
L}{\partial\partial_{i}\partial_{j_{1}}\cdots\partial_{j_{N-4}}\partial
_{k_{1}}\partial_{0}^{2}\Phi_{a}(\xi)},\ldots,
\end{align*}%
\begin{align*}
K_{2}^{aik_{1}\cdots k_{N-4}} & =c_{N-1}^{2}\frac{\partial L}{\partial
\partial_{i}\partial_{k_{1}}\cdots\partial_{k_{N-4}}\partial_{0}^{2}\Phi
_{a}(\xi)}-c_{N}^{2}\partial_{j_{1}}\frac{\partial L}{\partial\partial
_{i}\partial_{j_{1}}\partial_{k_{1}}\cdots\partial_{k_{N-4}}\partial_{0}^{2}%
\Phi_{a}(\xi)}, \\
K_{2}^{aik_{1}\cdots k_{N-3}} & =c_{N}^{2}\frac{\partial L}{\partial
\partial_{i}\partial_{k_{1}}\cdots\partial_{k_{N-3}}\partial_{0}^{2}\Phi
_{a}(\xi)},\ldots,
\end{align*}

\begin{align}
K_{N-2}^{ai} & =c_{N-1}^{N-2}\frac{\partial L}{\partial\partial_{i}%
\partial_{0}^{N-2}\Phi_{a}(\xi)}-c_{N}^{N-2}\partial_{j_{1}}\frac{\partial L%
}{\partial\partial_{i}\partial_{j_{1}}\partial_{0}^{N-2}\Phi_{a}(\xi )}, 
\notag \\
K_{N-2}^{aik_{1}} & =c_{N}^{N-2}\frac{\partial L}{\partial\partial
_{i}\partial_{k_{1}}\partial_{0}^{N-2}\Phi_{a}(\xi)},K_{N-1}^{ai}=c_{N}^{N-1}%
\frac{\partial L}{\partial\partial_{i}\partial_{0}^{N-1}\Phi_{a}(\xi)}. 
\tag{8}
\end{align}
Hence%
\begin{align*}
\frac{\delta\Lambda}{\delta\Phi_{a}(\xi)} & =\frac{\partial L}{\partial
\Phi_{a}(\xi)}-\partial_{i_{1}}\frac{\partial L}{\partial\partial_{i_{1}}%
\Phi_{a}(\xi)} \\
& +-\cdots+(-1)^{N}\partial_{i_{1}}\cdots\partial_{i_{N}}\frac{\partial L}{%
\partial\partial_{i_{1}}\cdots\partial_{i_{N}}\Phi_{a}(\xi)}, \\
\frac{\delta\Lambda}{\delta\partial_{0}\Phi_{a}(\xi)} & =\frac{\partial L}{%
\partial\partial_{0}\Phi_{a}(\xi)}-C_{2}^{1}\partial_{i_{1}}\frac{\partial L%
}{\partial\partial_{i_{1}}\partial_{0}\Phi_{a}(\xi)} \\
& +-\cdots+(-1)^{N-1}C_{N}^{1}\partial_{i_{1}}\cdots\partial_{i_{N-1}}\frac{%
\partial L}{\partial\partial_{i_{1}}\cdots\partial_{i_{N-1}}\partial
_{0}\Phi_{a}(\xi)}, \\
\frac{\delta\Lambda}{\delta\partial_{0}^{2}\Phi_{a}(\xi)} & =\frac{\partial L%
}{\partial\partial_{0}^{2}\Phi_{a}(\xi)}-C_{3}^{2}\partial_{i_{1}}\frac{%
\partial L}{\partial\partial_{i_{1}}\partial_{0}^{2}\Phi_{a}(\xi)} \\
& +-\cdots+(-1)^{N-2}C_{N}^{2}\partial_{i_{1}}\cdots\partial_{i_{N-2}}\frac{%
\partial L}{\partial\partial_{i_{1}}\cdots\partial_{i_{N-2}}\partial
_{0}^{2}\Phi_{a}(\xi)}, \\
& \ldots
\end{align*}

\begin{align}
\frac{\delta\Lambda}{\delta\partial_{0}^{N-1}\Phi_{a}(\xi)} & =\frac {%
\partial L}{\partial\partial_{0}^{N-1}\Phi_{a}(\xi)}-C_{N}^{N-1}%
\partial_{i_{1}}\frac{\partial L}{\partial\partial_{i_{1}}\partial_{0}^{N-1}%
\Phi_{a}(\xi)},  \notag \\
\frac{\delta\Lambda}{\delta\partial_{0}^{N}\Phi_{a}(\xi)} & =\frac{\partial L%
}{\partial\partial_{0}^{N}\Phi_{a}(\xi)}.  \tag{9}
\end{align}

Consider the difference between actions of two possible movements close to
each other.

\begin{align*}
\delta A &
=\int_{t_{0}}^{t}d\xi^{0}\Lambda\lbrack\Phi|_{\xi^{0}},\partial_{0}\Phi|_{%
\xi^{0}},\ldots,\partial_{0}^{N}\Phi|_{\xi^{0}}]=\int_{t_{0}}^{t}d\xi^{0}%
\int d^{3}\xi L(\Phi(\xi),\partial\Phi(\xi ),\ldots,\partial^{N}\Phi(\xi)) \\
& =\int_{t_{0}}^{t}d\xi^{0}\int d^{3}\xi\lbrack\frac{\delta\Lambda}{%
\delta\Phi_{a}(\xi)}\delta\Phi_{a}(\xi)+\frac{\delta\Lambda}{\delta
\partial_{0}\Phi_{a}(\xi)}\delta\partial_{0}\Phi_{a}(\xi)+\frac{%
\delta\Lambda }{\delta\partial_{0}^{2}\Phi_{a}(\xi)}\delta\partial_{0}^{2}%
\Phi_{a}(\xi) \\
& +\cdots+\frac{\delta\Lambda}{\delta\partial_{0}^{N}\Phi_{a}(\xi)}%
\delta\partial_{0}^{N}\Phi_{a}(\xi)]+\int_{t_{0}}^{t}d\xi^{0}\int
d^{3}\xi\partial_{i}\{[K^{ai}\delta\Phi_{a}(\xi)+K^{aik_{1}}\delta%
\partial_{k_{1}}\Phi_{a}(\xi) \\
& +K^{aik_{1}k_{2}}\delta\partial_{k_{1}}\partial_{k_{2}}\Phi_{a}(\xi
)+\cdots+K^{aik_{1}\cdots k_{N-1}}\delta\partial_{k_{1}}\cdots\partial
_{k_{N-1}}\Phi_{a}(\xi)+K^{aik_{1}\cdots
k_{N}}\delta\partial_{k_{1}}\cdots\partial_{k_{N}}\Phi_{a}(\xi)] \\
& +[K_{1}^{ai}\delta\partial_{0}\Phi_{a}(\xi)+K_{1}^{aik_{1}}\delta
\partial_{k_{1}}\partial_{0}\Phi_{a}(\xi)+K_{1}^{aik_{1}k_{2}}\delta
\partial_{k_{1}}\partial_{k_{2}}\partial_{0}\Phi_{a}(\xi)+\cdots \\
& +K_{1}^{aik_{1}\cdots k_{N-2}}\delta\partial_{k_{1}}\cdots\partial
_{k_{N-2}}\partial_{0}\Phi_{a}(\xi)]+[K_{2}^{ai}\delta\partial_{0}^{2}\Phi
_{a}(\xi)+K_{2}^{aik_{1}}\delta\partial_{k_{1}}\partial_{0}^{2}\Phi_{a}(\xi)
\\
&
+K_{2}^{aik_{1}k_{2}}\delta\partial_{k_{1}}\partial_{k_{2}}\partial_{0}^{2}%
\Phi_{a}(\xi)+\cdots+K_{2}^{aik_{1}\cdots
k_{N-3}}\delta\partial_{k_{1}}\cdots\partial_{k_{N-3}}\partial_{0}^{2}%
\Phi_{a}(\xi)]+\cdots \\
&
+[K_{N-2}^{ai}\delta\partial_{0}^{N-2}\Phi_{a}(\xi)+K_{N-2}^{aik_{1}}\delta%
\partial_{k_{1}}\partial_{0}^{N-2}\Phi_{a}(\xi)]+K_{N-1}^{ai}\delta%
\partial_{0}^{N-1}\Phi_{a}(\xi)\}
\end{align*}%
\begin{align*}
& =\int_{t_{0}}^{t}d\xi^{0}\int d^{3}\xi\lbrack\frac{\delta\Lambda}{%
\delta\Phi_{a}(\xi)}-\partial_{0}\frac{\delta\Lambda}{\delta\partial_{0}%
\Phi_{a}(\xi)}+\partial_{0}^{2}\frac{\delta\Lambda}{\delta\partial_{0}^{2}%
\Phi_{a}(\xi)}+-\cdots \\
& +(-1)^{N}\partial_{0}^{N}\frac{\delta\Lambda}{\delta\partial_{0}^{N}%
\Phi_{a}(\xi)}]\delta\Phi_{a}(\xi)+\int_{t_{0}}^{t}d\xi^{0}\int
d^{3}\xi\partial_{\lambda}[B^{a\lambda}\delta\Phi_{a}(\xi)+B^{a\lambda%
\nu_{1}}\delta\partial_{\nu_{1}}\Phi_{a}(\xi)+
\end{align*}%
\begin{align}
&
+B^{a\lambda\nu_{1}\nu_{2}}\delta\partial_{\nu_{1}}\partial_{\nu_{2}}%
\Phi_{a}(\xi)+\cdots+B^{a\lambda\nu_{1}\cdots\nu_{N-1}}\delta\partial_{\nu
_{1}}\cdots\partial_{\nu_{N-1}}\Phi_{a}(\xi)  \notag \\
& +B^{a\lambda\nu_{1}\cdots\nu_{N}}\delta\partial_{\nu_{1}}\cdots
\partial_{\nu_{N}}\Phi_{a}(\xi)]  \tag{10}
\end{align}
Using the least action principle, one re-obtains\textbf{\ }the\textbf{\
Euler-Lagrange equation }

\begin{equation}
\frac{\delta\Lambda}{\delta\Phi_{a}(\xi)}-\partial_{0}\frac{\delta\Lambda }{%
\delta\partial_{0}\Phi_{a}(\xi)}+\partial_{0}^{2}\frac{\delta\Lambda}{%
\delta\partial_{0}^{2}\Phi_{a}(\xi)}+-\cdots+(-1)^{N}\partial_{0}^{N}\frac{%
\delta\Lambda}{\delta\partial_{0}^{N}\Phi_{a}(\xi)}=0  \tag{11}
\end{equation}
Noting eqn.(9), one easily sees that eqns.(11) and (5) are exactly the same.

Let 
\begin{align}
\frac{\delta\Lambda}{\delta\partial_{0}\Phi_{a}(\xi)} & =\pi_{1}^{a}(\xi),%
\frac{\delta\Lambda}{\delta\partial_{0}^{2}\Phi_{a}(\xi)}=\pi_{2}^{a}(\xi),%
\ldots,\frac{\delta\Lambda}{\delta\partial_{0}^{N}\Phi_{a}(\xi)}%
=\pi_{N}^{a}(\xi),  \notag \\
H & =\int d^{3}\xi\lbrack\pi_{1}^{a}(\xi)\partial_{0}\Phi_{a}(\xi)+\pi
_{2}^{a}(\xi)\partial_{0}^{2}\Phi_{a}(\xi)+\cdots+\pi_{N}^{a}(\xi)\partial
_{0}^{N}\Phi_{a}(\xi)]-\Lambda  \tag{12}
\end{align}

One easily gets

\begin{equation*}
\delta H=\int
d^{3}\xi\lbrack\partial_{0}\Phi_{a}(\xi)\delta\pi_{1}^{a}(\xi)+%
\partial_{0}^{2}\Phi_{a}(\xi)\delta\pi_{2}^{a}(\xi)+\cdots+\partial
_{0}^{N}\Phi_{a}(\xi)\delta\pi_{N}^{a}(\xi)-\frac{\delta\Lambda}{\delta
\Phi_{a}(\xi)}\delta\Phi_{a}(\xi)]
\end{equation*}%
\begin{align}
& -\int d\sigma_{i}\{[K^{ai}\delta\Phi_{a}(\xi)+K^{aik_{1}}\delta
\partial_{k_{1}}\Phi_{a}(\xi)+K^{aik_{1}k_{2}}\delta\partial_{k_{1}}%
\partial_{k_{2}}\Phi_{a}(\xi)  \notag \\
& +\cdots+K^{aik_{1}\cdots k_{N-1}}\delta\partial_{k_{1}}\cdots
\partial_{k_{N-1}}\Phi_{a}(\xi)]  \notag \\
& +[K_{1}^{ai}\delta\partial_{0}\Phi_{a}(\xi)+K_{1}^{aik_{1}}\delta
\partial_{k_{1}}\partial_{0}\Phi_{a}(\xi)+K_{1}^{aik_{1}k_{2}}\delta
\partial_{k_{1}}\partial_{k_{2}}\partial_{0}\Phi_{a}(\xi)  \notag \\
& +\cdots+K_{1}^{aik_{1}\cdots k_{N-2}}\delta\partial_{k_{1}}\cdots
\partial_{k_{N-2}}\partial_{0}\Phi_{a}(\xi)]  \notag \\
&
+[K_{2}^{ai}\delta\partial_{0}^{2}\Phi_{a}(\xi)+K_{2}^{aik_{1}}\delta%
\partial_{k_{1}}\partial_{0}^{2}\Phi_{a}(\xi)+K_{2}^{aik_{1}k_{2}}\delta%
\partial_{k_{1}}\partial_{k_{2}}\partial_{0}^{2}\Phi_{a}(\xi)  \notag \\
& +\cdots+K_{2}^{aik_{1}\cdots k_{N-3}}\delta\partial_{k_{1}}\cdots
\partial_{k_{N-3}}\partial_{0}^{2}\Phi_{a}(\xi)]  \notag \\
&
+\cdots+[K_{N-2}^{ai}\delta\partial_{0}^{N-2}\Phi_{a}(\xi)+K_{N-2}^{aik_{1}}%
\delta\partial_{k_{1}}\partial_{0}^{N-2}\Phi_{a}(\xi)]  \notag \\
& +K_{N-1}^{ai}\delta\partial_{0}^{N-1}\Phi_{a}(\xi)\}.  \tag{13}
\end{align}
This suggests that the Hamiltonian $H$ is a functional of $\{\Phi_{a},\pi
_{1}^{a},\ldots,\pi_{N}^{a}\}$

\begin{equation}
H=H[\xi^{0},\Phi,\pi_{1},\pi_{2},\ldots,\pi_{N}],  \tag{14}
\end{equation}

\begin{equation}
\frac{\delta H}{\delta\Phi_{a}(\xi)}=-\frac{\delta\Lambda}{%
\delta\Phi_{a}(\xi)},\frac{\delta H}{\delta\pi_{1}^{a}(\xi)}%
=\partial_{0}\Phi_{a}(\xi ),\frac{\delta H}{\delta\pi_{2}^{a}(\xi)}%
=\partial_{0}^{2}\Phi_{a}(\xi ),\ldots,\frac{\delta H}{\delta\pi_{N}^{a}(\xi)%
}=\partial_{0}^{N}\Phi_{a}(\xi),  \tag{15}
\end{equation}
and

\begin{equation}
\Lambda=\int_{\Delta}d^{3}\xi\lbrack\frac{\delta H}{\delta\pi_{1}^{a}(\xi)}%
\pi_{1}^{a}(\xi)+\frac{\delta H}{\delta\pi_{2}^{a}(\xi)}\pi_{2}^{a}(\xi)+%
\cdots+\frac{\delta H}{\delta\pi_{N}^{a}(\xi)}\pi_{N}^{a}(\xi)]-H  \tag{16}
\end{equation}
From the Euler-Lagrange equation (11), one gets

\begin{equation}
\partial_{0}\pi_{1}^{a}(\xi)-\partial_{0}^{2}\pi_{2}^{a}(\xi)-+\cdots
-(-1)^{N}\partial_{0}^{N}\pi_{N}^{a}(\xi)=-\frac{\delta H}{%
\delta\Phi_{a}(\xi)}.  \tag{17}
\end{equation}
Eqn.(17) and eqn.(18)

\begin{equation}
\partial_{0}\Phi_{a}(\xi)=\frac{\delta H}{\delta\pi_{1}^{a}(\xi)},\partial
_{0}^{2}\Phi_{a}(\xi)=\frac{\delta H}{\delta\pi_{2}^{a}(\xi)},\ldots
,\partial_{0}^{N}\Phi_{a}(\xi)=\frac{\delta H}{\delta\pi_{N}^{a}(\xi)}, 
\tag{18}
\end{equation}
constitute\textbf{\ the canonical equations}. Note that when $N=1$ ( all
pre-G.R. field theories belong to this case ), canonical equations (17),
(18) \ read

\begin{align}
\partial_{0}\pi_{1}^{a}(\xi) & =-\frac{\delta H}{\delta\Phi_{a}(\xi )}, 
\notag \\
\partial_{0}\Phi_{a}(\xi) & =\frac{\delta H}{\delta\pi_{1}^{a}(\xi)}. 
\tag{19}
\end{align}
And one has

\begin{equation}
\frac{d}{d\xi^{0}}H=\partial_{0}H.  \tag{20}
\end{equation}
When $N=2$ ( G.R. is this case ), the canonical equations\ read

\begin{align}
\partial_{0}\pi_{1}^{a}(\xi)-\partial_{0}^{2}\pi_{2}^{a}(\xi) & =-\frac{%
\delta H}{\delta\Phi_{a}(\xi)},  \notag \\
\partial_{0}\Phi_{a}(\xi) & =\frac{\delta H}{\delta\pi_{1}^{a}(\xi )}, 
\notag \\
\partial_{0}^{2}\Phi_{a}(\xi) & =\frac{\delta H}{\delta\pi_{2}^{a}(\xi)}. 
\tag{21}
\end{align}
\bigskip And one has

\begin{equation}
\frac{d}{d\xi^{0}}\{H-\int d^{3}\xi\lbrack\partial_{0}\Phi_{a}(\xi
)\partial_{0}\pi_{2}^{a}(\xi)]\}=\partial_{0}H.  \tag{22}
\end{equation}

\section{Noether's theorem}

\subsection{Proof of Noether's theorem for Lagrangians containing up to $N$%
-th derivatives of field}

Now we have to deal with two kinds of derivatives of $L(x,\Phi (x),\partial
\Phi (x),\partial ^{2}\Phi (x),\ldots ,\partial ^{N}\Phi (x))$ with respect
to coordinates, $\partial _{\sigma }L=\partial L/\partial x^{\sigma }$ and $%
\eth _{\sigma }L=\eth L/\eth x^{\sigma }$, relating to each other through
the following equation,%
\begin{align}
\frac{\partial L}{\partial x^{\sigma }}& =\frac{\eth L}{\eth x^{\sigma }}+%
\frac{\partial L}{\partial \Phi _{a}(x)}\partial _{\sigma }\Phi _{a}(x)+%
\frac{\partial L}{\partial \partial _{\lambda _{1}}\Phi _{a}(x)}\partial
_{\sigma }\partial _{\lambda _{1}}\Phi _{a}(x)  \notag \\
& +\cdots +\frac{\partial L}{\partial \partial _{\lambda _{1}}\cdots
\partial _{\lambda _{N}}\Phi _{a}(x)}\partial _{\sigma }\partial _{\lambda
_{1}}\cdots \partial _{\lambda _{N}}\Phi _{a}(x)  \tag{23}
\end{align}

\begin{theorem}
If the action of classical fields over every spacetime region $\Omega $
remains unchanged under the following $r-$parameter family of infinitesimal
transformation of coordinates and fields%
\begin{align}
x^{\lambda }& \longmapsto \widetilde{x}^{\lambda }=x^{\lambda }+\delta
x^{\lambda },  \notag \\
\Phi _{a}(x)& \longmapsto \widetilde{\Phi }_{a}(x)=\Phi _{a}(x)+\delta \Phi
_{a}(x),  \tag{24}
\end{align}%
then there exist $r$ conserved quantities.
\end{theorem}

\begin{proof}
From eqn. (24) one has

\begin{align*}
\delta d^{4}x& =(\partial _{\sigma }\delta x^{\sigma })d^{4}x, \\
\delta \partial _{\lambda }& =-(\partial _{\lambda }\delta x^{\sigma
})\partial _{\sigma }, \\
\delta \lbrack \partial _{\lambda _{1}}\Phi _{a}(x)]& =(\delta \partial
_{\lambda _{1}})\Phi _{a}(x)+\partial _{\lambda _{1}}\delta \Phi
_{a}(x)=\partial _{\lambda _{1}}\delta \Phi _{a}(x)-\partial _{\sigma }\Phi
_{a}(x)\partial _{\lambda _{1}}\delta x^{\sigma }, \\
\delta \lbrack \partial _{\lambda _{1}}\partial _{\lambda _{2}}\Phi
_{a}(x)]& =\partial _{\lambda _{1}}\partial _{\lambda _{2}}\delta \Phi
_{a}(x)-\partial _{\lambda _{2}}\partial _{\sigma }\Phi _{a}(x)\partial
_{\lambda _{1}}\delta x^{\sigma }-\partial _{\lambda _{1}}\partial _{\sigma
}\Phi _{a}(x)\partial _{\lambda _{2}}\delta x^{\sigma } \\
& -\partial _{\sigma }\Phi _{a}(x)\partial _{\lambda _{1}}\partial _{\lambda
_{2}}\delta x^{\sigma }, \\
& \ldots 
\end{align*}%
\newline
\begin{eqnarray}
\delta \lbrack \partial _{\lambda _{1}}\partial _{\lambda _{2}}\cdots
\partial _{\lambda _{N}}\Phi _{a}(x)] &=&\partial _{\lambda _{1}}\partial
_{\lambda _{2}}\cdots \partial _{\lambda _{N}}\delta \Phi
_{a}(x)-\tsum\limits_{1\leq i\leq N}\partial _{\sigma }\partial _{\lambda
_{1}}\cdots \underline{\partial _{\lambda _{i}}}\cdots \partial _{\lambda
_{N}}\Phi _{a}(x)\partial _{\lambda _{i}}\delta x^{\sigma }  \notag \\
&&-\tsum\limits_{1\leq i<,j\leq N}\partial _{\sigma }\partial _{\lambda
_{1}}\cdots \underline{\partial _{\lambda _{i}}}\cdots \underline{\partial
_{\lambda _{j}}}\cdots \partial _{\lambda _{N}}\Phi _{a}(x)\partial
_{\lambda _{i}}\partial _{\lambda _{j}}\delta x^{\sigma }  \notag \\
&&-\tsum\limits_{1\leq i<,j\leq k\leq N}\partial _{\sigma }\partial
_{\lambda _{1}}\cdots \underline{\partial _{\lambda _{i}}}\cdots \underline{%
\partial _{\lambda _{j}}}\cdots \underline{\partial _{\lambda _{k}}}\cdots
\partial _{\lambda _{N}}\Phi _{a}(x)\partial _{\lambda _{i}}\partial
_{\lambda _{j}}\partial _{\lambda _{k}}\delta x^{\sigma }-\cdots  \\
&&-\tsum\limits_{k}\partial _{\sigma }\partial _{\lambda _{k}}\Phi
_{a}(x)\partial _{\lambda _{1}}\cdots \underline{\partial _{\lambda _{k}}}%
\cdots \partial _{\lambda _{N}}\delta x^{\sigma }  \notag \\
&&-\partial _{\sigma }\Phi _{a}(x)\partial _{\lambda _{1}}\cdots \partial
_{\lambda _{N}}\delta x^{\sigma }  \TCItag{25}
\end{eqnarray}%
Substutute eqns.(24) and (25) into the following equation

\begin{align*}
\delta A& =\int_{\Omega }(\delta d^{4}x)L+\int_{\Omega }d^{4}x[\frac{\eth L}{%
\eth x^{\sigma }}\delta x^{\sigma }+\frac{\partial L}{\partial \Phi _{a}(x)}%
\delta \Phi _{a}(x)+\frac{\partial L}{\partial \partial _{\lambda _{1}}\Phi
_{a}(x)}\delta \partial _{\lambda _{1}}\Phi _{a}(x) \\
& +\cdots +\frac{\partial L}{\partial \partial _{\lambda _{1}}\cdots
\partial _{\lambda _{N}}\Phi _{a}(x)}\delta \partial _{\lambda _{1}}\cdots
\partial _{\lambda _{N}}\Phi _{a}(x)]
\end{align*}%
one gets%
\begin{align}
\delta A& =\int_{\Omega }d^{4}x[\frac{\partial L}{\partial \Phi _{a}(x)}%
-\partial _{\lambda _{1}}\frac{\partial L}{\partial \partial _{\lambda
_{1}}\Phi _{a}(x)}+-\cdots +(-1)^{N}\partial _{\lambda _{1}}\cdots \partial
_{\lambda _{N}}\frac{\partial L}{\partial \partial _{\lambda _{1}}\cdots
\partial _{\lambda _{N}}\Phi _{a}(x)}]\times  \notag \\
& (\delta \Phi _{a}(x)-\partial _{\sigma }\Phi _{a}(x)\delta x^{\sigma
})+\int_{\Omega }d^{4}x\partial _{\lambda }[L\delta _{\sigma }^{\lambda
}\delta x^{\sigma }+B^{a\lambda }(\delta \Phi _{a}(x)-\partial _{\sigma
}\Phi _{a}(x)\delta x^{\sigma })  \notag \\
& +B^{a\lambda \nu _{1}}\partial _{\nu _{1}}(\delta \Phi _{a}(x)-\partial
_{\sigma }\Phi _{a}(x)\delta x^{\sigma })+B^{a\lambda \nu _{1}\nu
_{2}}\partial _{\nu _{1}}\partial _{\nu _{2}}(\delta \Phi _{a}(x)-\partial
_{\sigma }\Phi _{a}(x)\delta x^{\sigma })  \notag \\
& +\cdots +B^{a\lambda \nu _{1}\cdots \nu _{N-1}}\partial _{\nu _{1}}\cdots
\partial _{\nu _{N-1}}(\delta \Phi _{a}(x)-\partial _{\sigma }\Phi
_{a}(x)\delta x^{\sigma })]  \tag{26}
\end{align}%
The first integral at rhs vanishes for real movement, hence the second
integral does too. One gets the following equation due to the arbitrariness
of $\Omega $.%
\begin{align}
0& =\partial _{\lambda }[L\delta _{\sigma }^{\lambda }\delta x^{\sigma
}+B^{a\lambda }(\delta \Phi _{a}(x)-\partial _{\sigma }\Phi _{a}(x)\delta
x^{\sigma })+B^{a\lambda \nu _{1}}\partial _{\nu _{1}}(\delta \Phi
_{a}(x)-\partial _{\sigma }\Phi _{a}(x)\delta x^{\sigma })  \notag \\
& +B^{a\lambda \nu _{1}\nu _{2}}\partial _{\nu _{1}}\partial _{\nu
_{2}}(\delta \Phi _{a}(x)-\partial _{\sigma }\Phi _{a}(x)\delta x^{\sigma
})+\cdots  \notag \\
& +B^{a\lambda \nu _{1}\cdots \nu _{N-1}}\partial _{\nu _{1}}\cdots \partial
_{\nu _{N-1}}(\delta \Phi _{a}(x)-\partial _{\sigma }\Phi _{a}(x)\delta
x^{\sigma })]  \tag{27}
\end{align}%
Noting that both $\delta x^{\sigma }$ and $\delta \Phi _{a}(x)$ depend on $r$
real parameters, one can consider eqn.(27) as $r$ conservation laws.
\end{proof}

\subsection{"Conservation law due to "coordinate shift" invariance}

In this subsections, we restrict our discussion to Lagrangians which do not
manifestly contain coordinates and is invariant under "coordinate shift". In
this case, the action (1) remains unchanged under the following "coordinate
shift".

\begin{equation}
\delta x^{\sigma}=\varepsilon^{\sigma},\text{ }\delta\Phi_{a}(x)=0.  \tag{28}
\end{equation}
In this case, eqn.(27) reads

\begin{equation}
\partial_{\lambda}\tau_{\sigma}^{\lambda}=0,  \tag{29}
\end{equation}
where

\begin{align}
\tau_{\sigma}^{\lambda} &
=B^{a\lambda}\partial_{\sigma}\Phi_{a}(x)+B^{a\lambda\nu_{1}}\partial_{%
\nu_{1}}\partial_{\sigma}\Phi_{a}(x)+B^{a\lambda\nu_{1}\nu_{2}}\partial_{%
\nu_{1}}\partial_{\nu_{2}}\partial_{\sigma}\Phi_{a}(x)  \notag \\
& +\cdots+B^{a\lambda\nu_{1}\cdots\nu_{N-1}}\partial_{\nu_{1}}\cdots
\partial_{\nu_{N-1}}\partial_{\sigma}\Phi_{a}(x)-L\delta_{\sigma}^{\lambda} 
\tag{30}
\end{align}
is usually called energy-momentum tensor. Notice that "coordinate shift"
eqn.(28) is not an invariant concept under general coordinate
transformation. This is easily seen from the active viewpoint of
transformation.\ This explains why $\tau_{\sigma}^{\lambda}$ in not a tensor
under general coordinate transformation. We will get back to this problem
later.

\section{Hamilton's principal functional and Hamilton-Jacobi's equation}

Let us consider the difference between actions over spacetime region $\Omega$
of two real movements close to each other. Using eqns.(2) and (5), one gets,
for real movements

\begin{align}
\delta A[\Phi] & =\int_{\partial\Omega}ds_{\lambda}[B^{a\lambda}\delta
\Phi_{a}(x)+B^{a\lambda\nu_{1}}\delta\partial_{\nu_{1}}\Phi_{a}(x)+B^{a%
\lambda \nu_{1}\nu_{2}}\delta\partial_{\nu_{1}}\partial_{\nu_{2}}\Phi_{a}(x)
\notag \\
&
+\cdots+B^{a\lambda\nu_{1}\cdots\nu_{N-1}}\delta\partial_{\nu_{1}}\cdots%
\partial_{\nu N-1}\Phi_{a}(x)]  \tag{31}
\end{align}
From eqn.(31), one sees that the action over a spacetime region $\Omega$ of
a real movement is determined by the closed hyper-surface $\partial\Omega$,
and $\Phi|_{\partial\Omega},$ $\partial\Phi|_{\partial\Omega},\ldots,$ $%
\partial^{N-1}\Phi|_{\partial\Omega}$. It will be called the generalized
Hamilton's principal functional and denoted by

\begin{equation}
S=S[\partial\Omega,\Phi|_{\partial\Omega},\partial\Phi|_{\partial\Omega
},\ldots,\partial^{N-1}\Phi|_{\partial\Omega}]  \tag{32}
\end{equation}
Re-write eqn.(31) as

\begin{align}
\delta S & =\int_{\partial\Omega}ds_{\lambda}[B^{a\lambda}\delta\Phi
_{a}(x)+B^{a\lambda\nu_{1}}\delta\partial_{\nu_{1}}\Phi_{a}(x)+B^{a\lambda
\nu_{1}\nu_{2}}\delta\partial_{\nu_{1}}\partial_{\nu_{2}}\Phi_{a}(x)  \notag
\\
&
+\cdots+B^{a\lambda\nu_{1}\cdots\nu_{N-1}}\delta\partial_{\nu_{1}}\cdots%
\partial_{\nu N-1}\Phi_{a}(x)]  \tag{33}
\end{align}

Note that when $\Phi_{a}|_{\partial\Omega}$ is given, only one of the four
derivatives $\partial_{\lambda}\Phi_{a}|_{\partial\Omega}$ $%
(\lambda=0,1,2,3) $ is independent; when $\partial_{\lambda}\Phi_{a}|_{%
\partial\Omega}$ is given, only one of the four derivatives $%
\partial_{\mu}\partial_{\lambda}\Phi _{a}|_{\partial\Omega}$ $(\mu=0,1,2,3)$
is independent; and so on. Thus for a given suffix $a$, only $N$ items from $%
\Phi_{a}|_{\partial\Omega,}\partial_{\lambda_{1}}\Phi_{a}|_{\partial\Omega},%
\ldots,\partial_{\lambda_{1}}\cdots\partial_{\lambda_{N-1}}\Phi_{a}|_{%
\partial\Omega}$ $(\lambda _{j}=0,1,2,3)$ are independent.

In order to formulate the generalized Hamilton-Jacobi's equation, one needs
a new type of functional derivative.

\begin{definition}
Let $\Sigma$ be a hypersurface in spacetime $M$, $\Psi$ a function defined
on $M$, and $F=F[\Sigma,\Psi|_{\Sigma}]$ a functional of $\Sigma$ and$\ \Psi
|_{\Sigma}$. The functional derivatives are defined as follows. If the
variation of $F$ can be written as%
\begin{equation}
\delta F[\Sigma,\Psi|_{\Sigma}]=\int_{\Sigma}ds_{\lambda}\left\{
Y[\Sigma,\Psi|_{\Sigma},x)_{\mu}^{\lambda}\delta\Sigma^{\mu}(x)+Z[\Sigma
,\Psi|_{\Sigma},x)^{\lambda}\delta\Psi(x)\right\}  \tag{34}
\end{equation}
then $Y[\Sigma,\Psi|_{\Sigma},x)_{\mu}^{\lambda}$ and $Z[\Sigma,\Psi|_{%
\Sigma },x)^{\lambda}$ are called the functional derivative of $F$ with
respect to $\Sigma^{\mu}(x)$ and $\Psi(x)$, and denoted by 
\begin{equation}
Y[\Sigma,\Psi|_{\Sigma},x)_{\mu}^{\lambda}=\left( \frac{\delta F}{%
\delta\Sigma^{\mu}(x)}\right) ^{\lambda},\text{ }Z[\Sigma,\Psi|_{\Sigma
},x)^{\lambda}=\left( \frac{\delta F}{\delta\Psi(x)}\right) ^{\lambda} 
\tag{35}
\end{equation}
respectively.
\end{definition}

Hence we have%
\begin{equation}
\delta F[\Sigma,\Psi|_{\Sigma}]=\int_{\Sigma}ds_{\lambda}\left[ \left( \frac{%
\delta F}{\delta\Sigma^{\mu}(x)}\right) ^{\lambda}\delta\Sigma^{\mu
}(x)+\left( \frac{\delta F}{\delta\Psi(x)}\right) ^{\lambda}\delta \Psi(x)%
\right]  \tag{36}
\end{equation}

The hypersurface $\Sigma$ is given by the parameter equation

\begin{equation}
x^{\mu}=\Sigma^{\mu}(\theta^{1},\theta^{2},\theta^{3})  \tag{37}
\end{equation}
The $\delta\Sigma^{\mu}(x)$ in eqn.(33) is

\begin{equation}
\delta\Sigma^{\mu}(x)=\widetilde{\Sigma}^{\mu}(\theta^{1},\theta^{2},%
\theta^{3})-\Sigma^{\mu}(\theta^{1},\theta^{2},\theta^{3}).  \tag{38}
\end{equation}

Now, from eqn.(33) we have

\begin{equation}
\left( \frac{\delta S}{\delta\Phi_{a}(x)}\right) ^{\lambda}=B^{a\lambda
},\left( \frac{\delta S}{\delta\partial_{\nu_{1}}\Phi_{a}(x)}\right)
^{\lambda}=B^{a\lambda\nu_{1}},\ldots,\left( \frac{\delta S}{\delta
\partial_{\nu_{1}}\cdots\partial_{\nu_{N-1}}\Phi_{a}(x)}\right) ^{\lambda
}=B^{a\lambda\nu_{1}\cdots\nu_{N-1}}.  \tag{39}
\end{equation}
Follow the evolution of one real movement and observe the change of its
action.

\begin{align}
\delta S & =\int_{\partial\Omega}ds_{\lambda}L\delta_{\sigma}^{\lambda
}\delta\Sigma^{\sigma}(x)  \notag \\
& =\int_{\partial\Omega}ds_{\lambda}[\left( \frac{\delta S}{\delta
\Sigma^{\sigma}(x)}\right) ^{\lambda}\delta\Sigma^{\sigma}(x)+\left( \frac{%
\delta S}{\delta\Phi_{a}(x)}\right) ^{\lambda}\partial_{\sigma}\Phi
_{a}(x)\delta\Sigma^{\sigma}(x)  \notag \\
& +\left( \frac{\delta S}{\delta\partial_{\nu_{1}}\Phi_{a}(x)}\right)
^{\lambda}\partial_{\sigma}\partial_{\nu_{1}}\Phi_{a}(x)\delta\Sigma^{\sigma
}(x)+\cdots  \notag \\
& +\left( \frac{\delta S}{\delta\partial_{\nu_{1}}\cdots\partial_{\nu_{N-1}}%
\Phi_{a}(x)}\right)
^{\lambda}\partial_{\sigma}\partial_{\nu_{1}}\cdots\partial_{\nu_{N-1}}%
\Phi_{a}(x)\delta\Sigma^{\sigma}(x)]  \tag{40}
\end{align}
From eqns.(30), (39) and (40), we get the generalized Hamilton-Jacobi's
equation.

\begin{equation}
\left( \frac{\delta S}{\delta\Sigma^{\sigma}(x)}\right)
^{\lambda}+\tau_{\sigma}^{\lambda}=0.\text{ }  \tag{41}
\end{equation}

\begin{remark}
So far we have presented a general variational priciple for classical
fields. The only postulate made in this formalism is the least action
principle. This formalism applies to all the classical fields $\{\Phi
_{a}(x)\}$ with a Lagrangian $L(x,\Phi (x),\partial \Phi (x),\partial
^{2}\Phi (x),\ldots ,\partial ^{N}\Phi (x))$, say, Newtonian fluid
mechanics, Maxwell's electromagnetic field, general relativity, etc. The
specific symmtries and covariance of a classical field are the heritage from
the Lagrangian, not from this general formalism. This formalism yields
manifestly Galilean (Lorentzian, general) covariant field theory when the
inputted Lagrangian is Galilean (Lorentzian, general) covariant. It is worth
noting that all the results obtained above, are in great harmony with each
other. We will apply this general variational priciple to general
relativity, especially apply the generalized Noether's theorem to the long
standing problem, conservation and non-conservation in curved spacetime in
part II and part III.
\end{remark}

\begin{acknowledgement}
I am grateful to Prof. Zhanyue Zhao, Prof. Shihao Chen and Prof. Xiaoning Wu
for helpful discussions.
\end{acknowledgement}

\end{document}